\documentclass[
 aip,
 sd,amssymb,amsmath,
 preprint
]{revtex4-1}
\usepackage{amsthm,color}
\usepackage{subfig}
\usepackage{graphicx}
\usepackage{appendix}
\usepackage{caption}
\usepackage{url}
\usepackage{epstopdf}
\usepackage{comment}
\usepackage{hyperref}

\newtheorem{theorem}{Theorem}[section]

\newtheorem{lemma}[theorem]{Lemma}

\newcommand\PG[1]{{\color{black}#1}}

\newcommand{\kb}[1]{{\color{black} \text{#1}}}

\begin{document}
\title{A mean-field game model for homogeneous flocking }
\author{Piyush Grover}
 \affiliation{ 
Mitsubishi Electric Research Labs, Cambridge, MA 02139, USA}
\email{grover@merl.com}
\author{Kaivalya Bakshi}%
 \affiliation{ 
Mitsubishi Electric Research Labs, Cambridge, MA 02139, USA}
 \affiliation{ 
Aerospace Engineering, Georgia Tech, Atlanta, GA 30332, USA}
 \author{Evangelos A. Theodorou}
 \affiliation{ 
Aerospace Engineering, Georgia Tech, Atlanta, GA 30332, USA}

\date{\today}

\begin{abstract}
Empirically derived continuum models of collective behavior among large populations of dynamic agents are a subject of intense study in several fields, including biology, engineering and finance. We formulate and study a mean-field game model whose behavior mimics an empirically derived nonlocal homogeneous flocking model for agents with gradient self-propulsion dynamics. The mean-field game framework provides a \emph{non-cooperative optimal control} description of the behavior of a population of agents in a distributed setting. In this description, each agent's state is driven by optimally controlled dynamics that result in a Nash equilibrium between itself and the population. The optimal control is computed by minimizing a cost that depends only on its own state, and a mean-field term. The agent distribution in phase space evolves under the optimal feedback control policy. We exploit the low-rank perturbative nature of the nonlocal term in the forward-backward system of equations governing the state and control distributions, and provide a closed-loop linear stability analysis demonstrating that our model exhibits bifurcations similar to those found in the empirical model. The present work is a step towards developing a set of tools for systematic analysis, and eventually \emph{design}, of collective behavior of non-cooperative dynamic agents via an inverse modeling approach. \end{abstract}

\maketitle
\begin{quotation}
\PG{While the analysis of emergent behavior in a large population of dynamic agents is a classical topic, the design of desired macroscopic behavior in such systems is a grand engineering challenge.} Such systems are often studied using continuum models, involving empirically derived systems of nonlinear partial differential equations that govern the distribution of agents in the phase space. The various terms in these equations represent intrinsic dynamics of the agents, mutual attraction and/or repulsion, and noise. An important class of such models concern flocking, both in nature, and engineering applications such as bio-inspired control of multi-agent robotics, traffic modeling, power-grid synchronization etc. We take a mean-field game approach to derive a control system that mimics the behavior of one such class of models in the setting of non-cooperative agents. A mean-field game is a coupled system of partial differential equations that govern the state and optimal control distributions of a representative agent in a Nash equilibrium with the population. Using a linear stability analysis, we recover phase transitions that have been observed in the corresponding empirical model, as well as find some new ones, as the control penalty is changed.
\end{quotation}
\section{Introduction}
Continuum models of large populations of interacting dynamic agents are popular in mathematical biology\cite{resat2009kinetic}, and also have been employed in numerous applications such as multi-agent robotics \cite{beni2004swarm}, finance \cite{bonabeau2002agent} and traffic modeling \cite{whitham1955kinematic}. The aim of such models is to accurately represent the macroscopic dynamics of the population, and its dependence on parameters. Typically, such models are derived by starting with an empirical dynamical system for a representative agent. This system typically involves the intrinsic dynamics of the agent, a coupling function\cite{stankovski2017coupling} describing its interaction with the population, and noise. From this single agent dynamical system, a continuum description is obtained by deriving a macroscopic equation for the distribution of agents in the phase space. We call this class of models \emph{uncontrolled}.

An alternative way of deriving continuum models of collective behavior is \kb{via} a corresponding variational principle. In this approach, the dynamical system for a representative agent includes its intrinsic dynamics, a control term and noise. The unknown control term is obtained as a solution to an optimization problem. Within this variational (or optimization) framework for large populations, there are multiple classes of modeling strategies \cite{nourian2013nash}. If one takes a centralized global optimization viewpoint, the corresponding problem is that of \emph{mean-field control}, i.e. it is assumed that each agent is being controlled by a central entity whose goal is to optimize a macroscopic cost function\cite{elamvazhuthi2016optimal} that includes interaction among the population. In a distributed setting, there is no central entity, and the agents can either be \emph{cooperative} or \emph{non-cooperative}. In the former case, each agent choses its control to optimize a global sum of cost functions of the population. 

On the other hand, in the \emph{non-cooperative} mean field setting that we are interested in, each agent optimizes only its \kb{individual} cost function. This cost function involves coupling with the population solely via a mean-field term. This is the setting of mean-field games (MFG)\cite{caines2015mean,lasry2007mean,huang2007large}. In this setting, a Hamilton-Jacobi-Bellman (HJB) equation (posed backward in time) characterizes the optimal feedback control for a representative agent under the assumption that the (cost) coupling function depends only on its own state, and possibly time. A Fokker-Planck (FP) equation governs the evolution of agent density in phase space. A consistency principle \cite{huang2007large} requires that the coupling function used in the agent HJB equation is reproduced as its own average over the continuum of agents. Under fairly general conditions, solutions to MFG model can be shown to possess $\epsilon$-Nash property, i.e., unilateral benefit of any deviation from the computed control policy by a single agent vanishes rapidly as the population becomes large.

\begin{figure}[t]
\includegraphics[width=2.65in,height=2.2in]{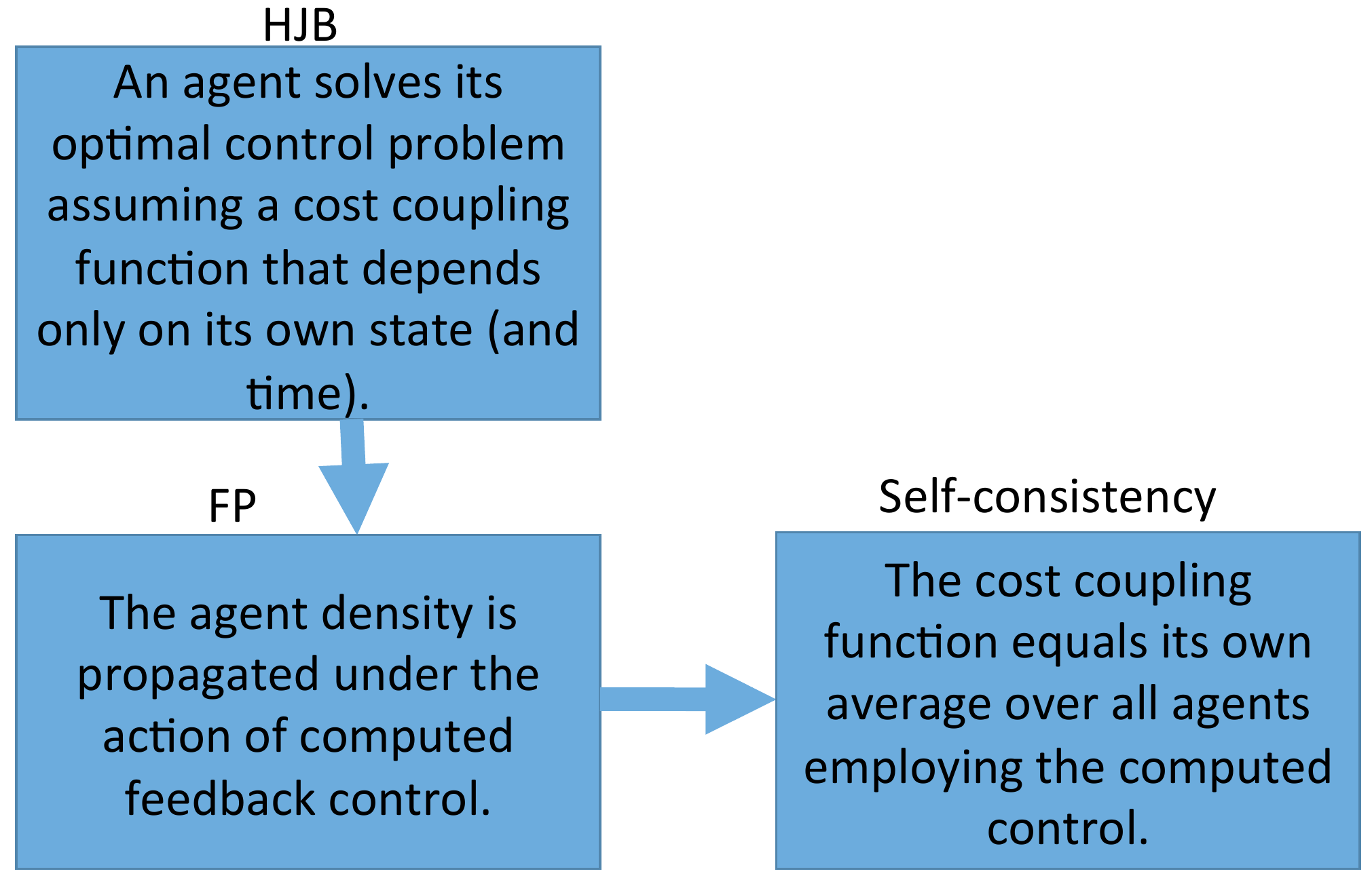}
\caption{The MFG framework}
 \label{fig:mfg}
 \end{figure}

The classical (uncontrolled) Cucker-Smale (CS) flocking model\cite{cucker2007emergent} \kb{describes} a system of finite population of coupled agents with trivial intrinsic dynamics, moving solely under the influence of an alignment force, and noise. This was followed by several continuum descriptions\cite{ha2009simple, choi2017emergent}, and was recently generalized to a continuum model with self-propulsion effects in the homogeneous case \cite{barbaro2016phase} (i.e., assuming spatial homogeneity). This latter generalization results in existence of non-zero mean velocity distribution resulting from symmetry breaking, a wide range of `disordered' states consisting of multiple flocks, and other phase transitions.

A MFG model for a continuum of coupled Kuramoto oscillators\cite{kuramoto1975self} was described in a seminal work \cite{yin2012synchronization} that influences the development in the current paper. Building upon this work, a MFG model for the classical inhomogeneous CS was then proposed\cite{nourian2011mean}; the stability analysis was partially addressed. This was followed by a homogeneous flocking MFG model for coupled agents with trivial intrinsic dynamics, along with linear and nonlinear stability analysis\cite{nourian2014mean}.  Also of interest is an approach \cite{degond2014large} where agents apply a gradient descent rather than solve an HJB equation, since the Nash equilibria of the MFG are recovered under certain conditions using this approach. 

\PG{The contributions of this paper are as follows. We formulate a MFG model for homogeneous flocking of agents driven by self-propulsion and noise. In contrast to the earlier work on homogeneous MFG model with trivial intrinsic dynamics \cite{nourian2014mean}, this model exhibits phase transitions (bifurcations) that mimic those present in the corresponding uncontrolled model \cite{barbaro2016phase}. We generalize the stability analysis developed in previous MFG models \cite{yin2012synchronization, nourian2014mean,huang2007large,gueant2009reference} to agents with gradient nonlinear dynamics, and employ a method used to study reaction-diffusion equations\cite{anastasio2017geometric} to derive a semi-analytical stability criterion. Besides qualitatively explaining the phase transition phenomena, quantitative results useful in control design are obtained from the numerical analysis. Decreasing the control control penalty below a threshold causes the zero mean velocity steady state of the MFG model to lose stability via pitchfork bifurcation \cite{Wiggins1990}. This results in a pair of stable steady states with non-zero mean velocity. If the control is made even cheaper, a new stable regime (nonexistent in the uncontrolled model) emerges for zero mean velocity steady states in the small noise case via a subcritical pitchfork bifurcation.

}

\section{Uncontrolled formulation}\label{sec:uncontrolled}
\vspace{-0.35cm}
We \kb{briefly review} here the uncontrolled formulation from Ref. \onlinecite{barbaro2016phase} \kb{which} provides a homogeneous model for CS flocking with self-propulsion. Consider a population of $N$ agents moving in phase space \kb{($(q,p)\in\mathbb{R}^2$)}, where each agent is acted upon by a gradient self-propulsion term, a CS coupling force with localization kernel $K$ in position space that aligns the agents' velocity with the neighbors, and noise. The \kb{dynamics} for $i$th agent are
\begin{eqnarray*}
dq_i=p_idt,\\
dp_i=a(p_i)dt+F(q_i,p_i,q_{-i},p_{-i})dt+\sigma d\omega_i,\label{eq:unc_v}
\end{eqnarray*}
where $a(p_i)=-\partial_pU(p_i)$, \kb{$U(p_i)=\alpha(\frac{p_i^4}{4}-\frac{p_i^2}{2})$}, $F(q_i,p_i,q_{-i},p_{-i})=\dfrac{1}{N}\dfrac{\sum_{j=1}^{N}  K(q_i,q_j)(p_j-p_i)}{\sum_{j=1}^{N} K(q_i,q_j)},$ $\sigma>0$ is the noise intensity, $\alpha>0$ defines the strength of the self-propulsion term, $K(q,q') = K(q',q) \geq 0,$ and  $K(q,q) = 1$, $q_{-i}=\{q_1,\dots ,q_{i-1}, q_{i+1}\dots\}$,  $p_{-i}=\{p_1,\dots ,p_{i-1}, p_{i+1}\dots\}$.

In the continuum limit ($N\rightarrow\infty$), the agent density $f(q,p,t)$ in phase space is \kb{governed by}
\begin{eqnarray*}\partial_tf+\partial_q(pf)+\partial_p(a(p)f+F[f]f)=\dfrac{\sigma^2}{2}\partial_{pp}f,\end{eqnarray*}
where \kb{$F[f](q,p,t)=(\bar{p}-p)$  and,} 
\begin{eqnarray*}\kb{$\bar{p}(q,t)=\dfrac{\int\int K(q,q')\:p\:f(q',p,t)dq'dp}{\int\int K(q,q')f(q',p,t)dq'dp}.$}\end{eqnarray*}
\PG{We denote the action of the operator $F$ on a function $f$ by $F[f](.)$.}
%
%
From here onwards, we consider the homogeneous case by dropping dependence on $q$, and use $x$ to denote the velocity $p$.
The uncontrolled \kb{dynamics} for the velocity of agent $i$ are
\vspace{-0.35cm}
\begin{flalign}
dx_i=a(x_i)dt+\dfrac{1}{N} \sum_{j=1}^{N}  (x_j-x_i)dt+\sigma d\omega_i,\label{eq:unc_x}
\end{flalign} 
\kb{with} corresponding density evolution
\vspace{-0.35cm}
\begin{equation}
\partial_tf=\partial_x(\alpha(x^2-1)xf+(x-\bar{x})f)+\frac{\sigma^2}{2}\partial_{xx}f, \label{eq:kinetichom}
\end{equation}
where $\bar{x}(t)=\dfrac{\int xf(x,t)dx}{\int f(x,t)dx}$.

%
\subsection{Fixed Points and Stability Analysis}\label{sec:stab_fp}
It is known \cite{tugaut2014phase,barbaro2016phase} that fixed points of Eq. (\ref{eq:kinetichom}) are \kb{given by}
\begin{flalign}
&f_{\infty}(x;\mu)=\frac{1}{Z}\exp\left(\frac{-2}{\sigma^2}\left[\alpha \frac{x^4}{4}+(1-\alpha)\frac{x^2}{2}-\mu x\right]\right), \label{eq:solhom}
\end{flalign}
where $\mu \in \mathbb{R}$ is the mean of the distribution, and $Z$ is the normalization \kb{factor}. For all positive values of parameters \kb{$(\sigma,\alpha)$}, the zero mean velocity solution $f_{\infty}(\cdot,0)$ always exists. For a range of parameters, two additional stable non-zero mean velocity solutions are created via a supercritical bifurcation, resulting in loss of stability of the zero mean solution. In Ref. \onlinecite{barbaro2016phase}, these stability properties were inferred numerically \kb{by} a Monte-Carlo approach.

We take a different approach, and consider the spectral stability of steady state solutions of Eq. (\ref{eq:kinetichom}).  In addition to gaining additional insight into the properties of the uncontrolled system, this also sets the stage for stability analysis of the MFG system in the next section.
 We consider perturbations of the form $f(x,t)=f_{\infty}(x)(1+\epsilon\tilde f(x,t))$. Then, the linearization of Eq. (\ref{eq:kinetichom}) is 
\begin{eqnarray*}
\partial_t\tilde f(x,t)=L[\tilde f](x,t)=L_{loc}[\tilde f](x,t)+L_{nonloc}[\tilde f](x,t),
\end{eqnarray*}
where, $\hat U(x)=U(x)+x^2/2-\mu x$, 
\begin{eqnarray*}
L_{loc}[\tilde f](x,t)=-\partial_x\hat U(x)\partial_x\tilde f(x,t)+(\sigma^2/2)\partial_{xx}\tilde f(x,t)
\end{eqnarray*}
is a local linear operator, and 
\begin{eqnarray*}
L_{nonloc}[\tilde f](x,t)=\frac{2}{\sigma^2}\partial_x\hat U(x)\int y \tilde f(y,t)f_{\infty}(y)dy
\end{eqnarray*}
is a \emph{nonlocal} linear operator. \PG{An operator $O$ is called nonlocal if $O[f](x_1)$ depends on $f(x_2)$ (or the derivatives $\partial_x f(x_2), \partial_{xx} f(x_2)$) for some $x_2\neq x_1$, and local otherwise.}
 Let $q(x)\equiv\frac{2}{\sigma^2}\partial_x\hat U(x)$. Then, $\partial_xf_{\infty}(x)=-q(x)f_{\infty}(x)$.
We define a Hilbert space $\mathbb{H}=L^2(\mathbb{R},f_{\infty}dx)$\PG{, i.e., the $f_\infty$-weighted inner-product space of square-integrable functions on the real line}. Then we can write a general form of the full linearized operator as
\begin{align}
L[\tilde f](x,t)=L_{loc}[\tilde f](x,t)+s_1(x)\langle g_1(\cdot),\tilde f(\cdot,t)\rangle\label{eq:Lgen},
\end{align}
where $s_1(x)=q(x), g_1(x)=x$ for our case, and the inner product is understood to be $\langle\cdot,\cdot\rangle_\mathbb{H}$.
We note that $L_{loc}$ is a self-adjoint operator\cite{pavliotis2014stochastic} on $\mathbb{H}$ which has a non-positive discrete real spectrum of the form $0=\lambda_1>\lambda_2>\lambda_3\dots$. It has a complete set of orthogonal eigenfunctions $\{\xi_i(x)\}_{i\in\mathbb{N}}$. The first eigenfunction $\xi_1$, spanning the kernel of $L_{loc}$, is a constant function. Following the approach presented in Refs. \onlinecite{anastasio2017geometric,frettas1995stability} for nonlocal eigenvalue problems in reaction-diffusion equations \PG{(also see Ref. \onlinecite{kapitula2012stability})}, we consider the following eigenvalue problem   
\begin{flalign}
\lambda w=L_{loc}w&+s_1(x)\langle g_1,w\rangle\implies\nonumber\\
&0=(L_{loc}-\lambda I)w+s_1(x)\langle g_1,w\rangle.&\label{eq:eval}
\end{flalign}
Note that an eigenfunction $w$ of $L$ satisfying $\langle w,g_1\rangle=0$ is also an eigenfunction of $L_{loc}$, i.e. $w=v_i$ for some $i$ with eigenvalue $\lambda=\lambda_i$. We search for   eigenfunctions such that $\langle w,g_1\rangle$ is nonzero. The corresponding eigenvalues are called `moving' eigenvalues in Ref. \onlinecite{frettas1995stability}. Multiplying both sides of Eq. (\ref{eq:eval}) with the resolvent $R_{\lambda}=(L_{loc}-\lambda I)^{-1}$,
\begin{eqnarray}
0=w+R_{\lambda}s_1(x)\langle g_1,w\rangle.\nonumber
\end{eqnarray}
\PG{Taking the inner product of the above equation with $g_1$,}
\begin{eqnarray}
0=\langle g_1,w\rangle+\langle R_{\lambda}s_1(x),g_1\rangle\langle g_1,w\rangle.\label{eq:eval1}
\end{eqnarray}

\kb{For} an arbitrary function $z(x), R_{\lambda}z=\sum_{i=1}^{\infty}\dfrac{\langle \xi_i,z\rangle}{\lambda_i-\lambda}\xi_i$. Evaluating the inner product in Eqs. \ref{eq:eval1},
\begin{equation}\langle R_{\lambda}s_1(x),g_1\rangle=\langle R_{\lambda}q(x),x\rangle=\sum_{i=2}^{\infty}\dfrac{\langle \xi_i,q(x)\rangle}{\lambda_i-\lambda}\langle \xi_i,x\rangle.\end{equation}
Using this result in Eq. (\ref{eq:eval1}),
\begin{eqnarray}
\langle w,x\rangle (1+\sum_{i=2}^{\infty}\dfrac{\langle \xi_i,q(x)\rangle}{\lambda_i-\lambda}\langle \xi_i,x\rangle)=0. \nonumber
\end{eqnarray}
Hence, either $\langle w,x\rangle=0$, or $1+\sum_{i=2}^{\infty}\dfrac{\langle \xi_i,q(x)\rangle}{\lambda_i-\lambda}\langle \xi_i,x\rangle=0$. But we are looking for moving eigenvalues, i.e. $w$ s.t. $\langle w,x\rangle\neq 0$, hence the eigenvalue equation reduces to:
\begin{eqnarray} 
h(\lambda)\equiv 1+\sum_{i=2}^{\infty}\dfrac{\langle \xi_i,q(x)\rangle}{\lambda_i-\lambda}\langle \xi_i,x\rangle=0.\label{eq:eval3}
\end{eqnarray}

A sufficient condition for Eq. (\ref{eq:eval3}) to have only real roots is that the function $h(\lambda)$ is Herglotz, or equivalently, the product $\langle \xi_i,q(x)\rangle\langle \xi_i,x\rangle$ has \kb{the} same sign for all $i$. Using integration by parts on eigenvalue equation for $L_{loc}$, one can show that  $\langle \xi_i,x\rangle=-\dfrac{\sigma^2}{2\lambda_i}\langle \xi_i,q(x)\rangle$. Thus the Herglotz condition is satisfied since $\lambda_i< 0$ for all $i>1$.

\begin{figure*}[t]
\subfloat[]{\includegraphics[width=1.9in,height=2in]{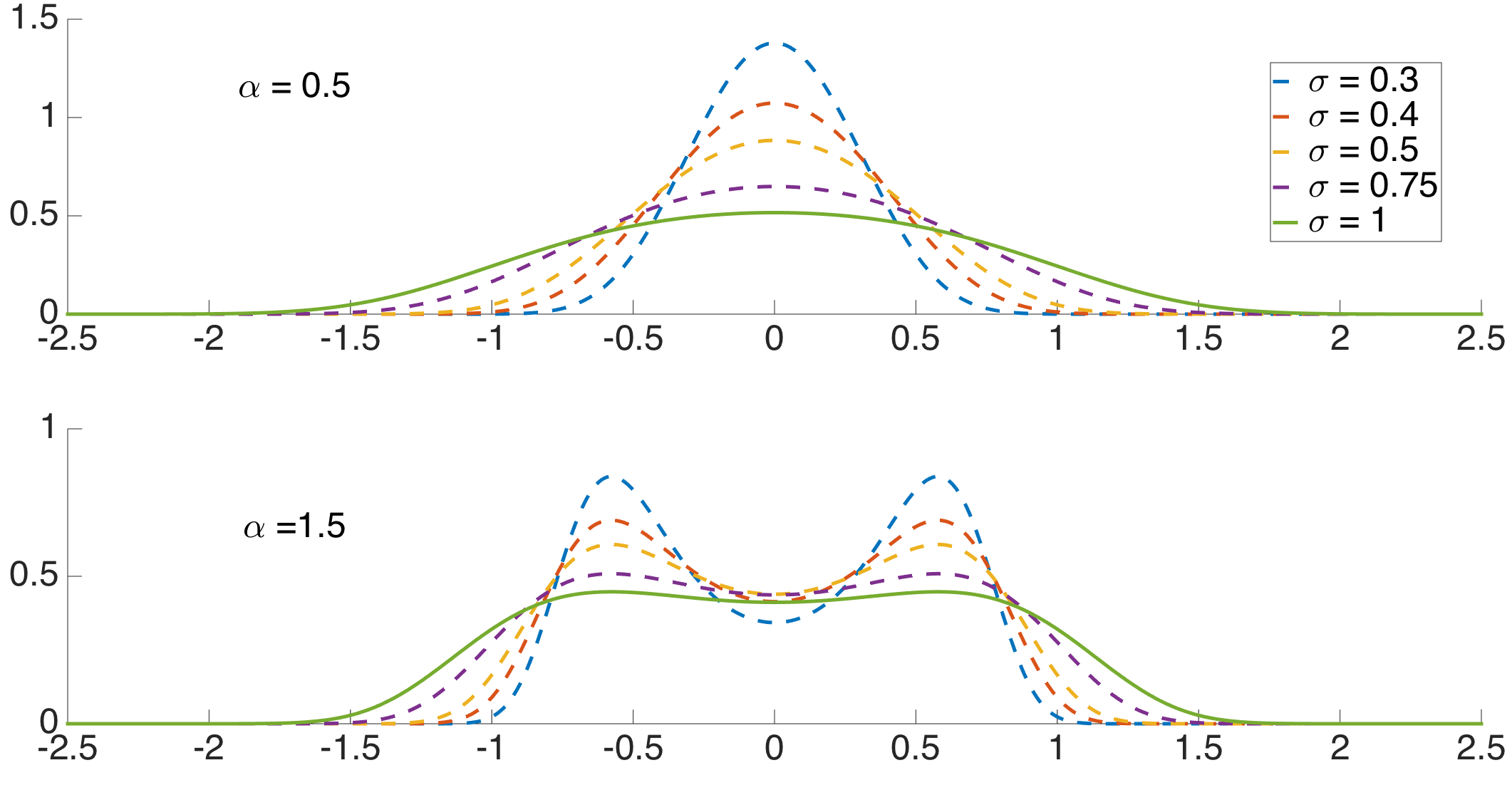}	}
\subfloat[]{\includegraphics[width=1.3in,height=1in]{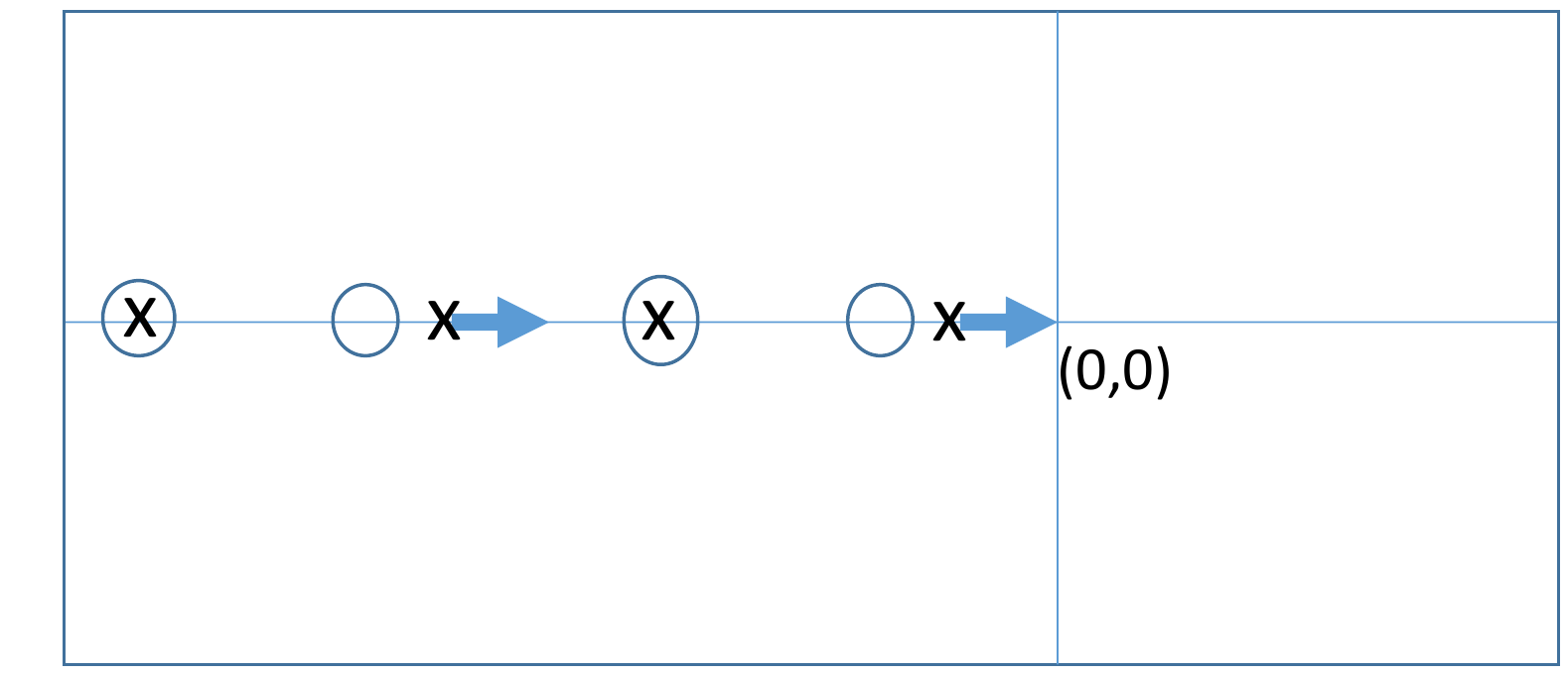}}
\subfloat[]{\includegraphics[width=1.9in,height=2in]{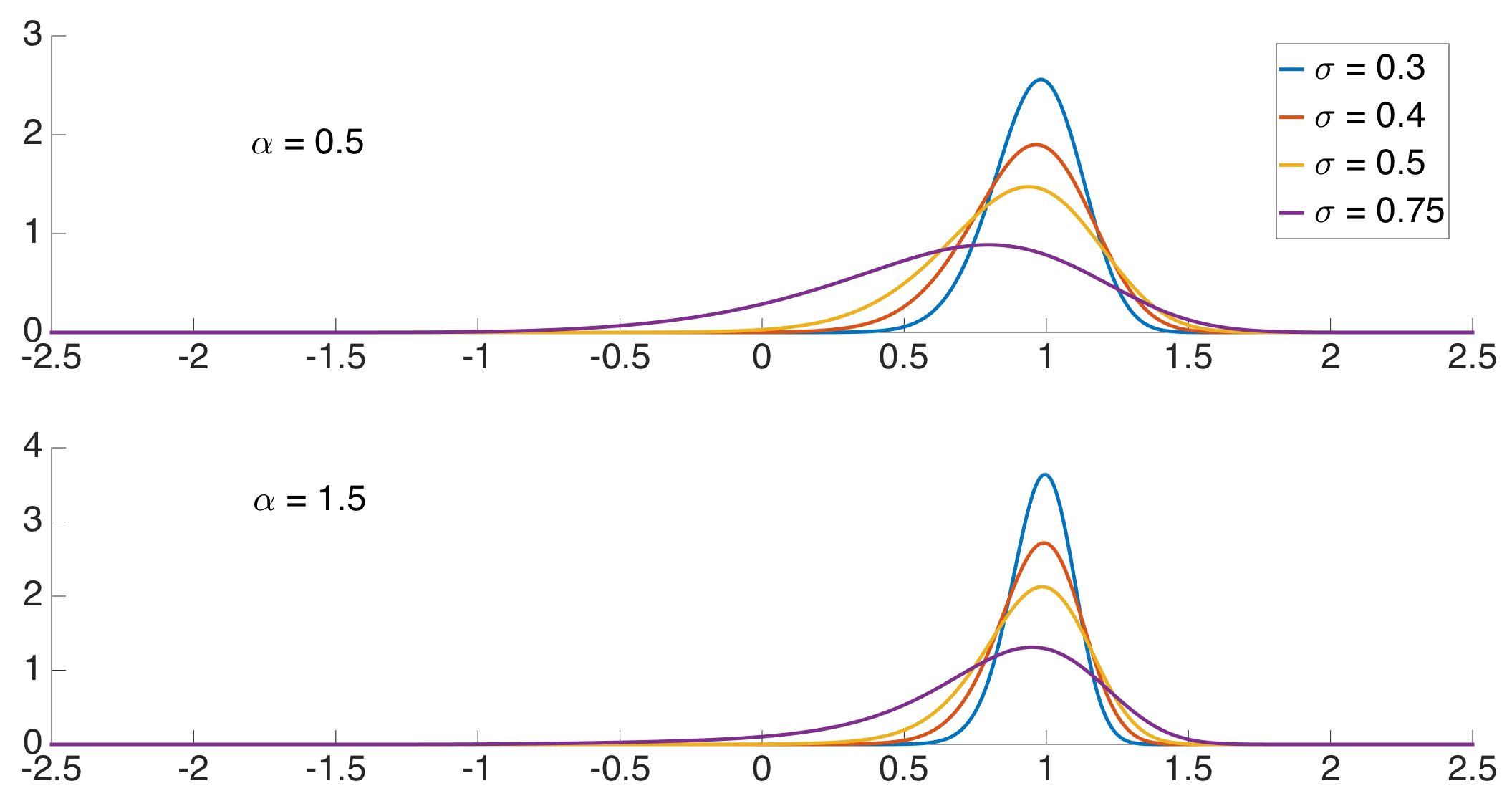}}
\subfloat[]{\includegraphics[width=2.3in,height=2in]{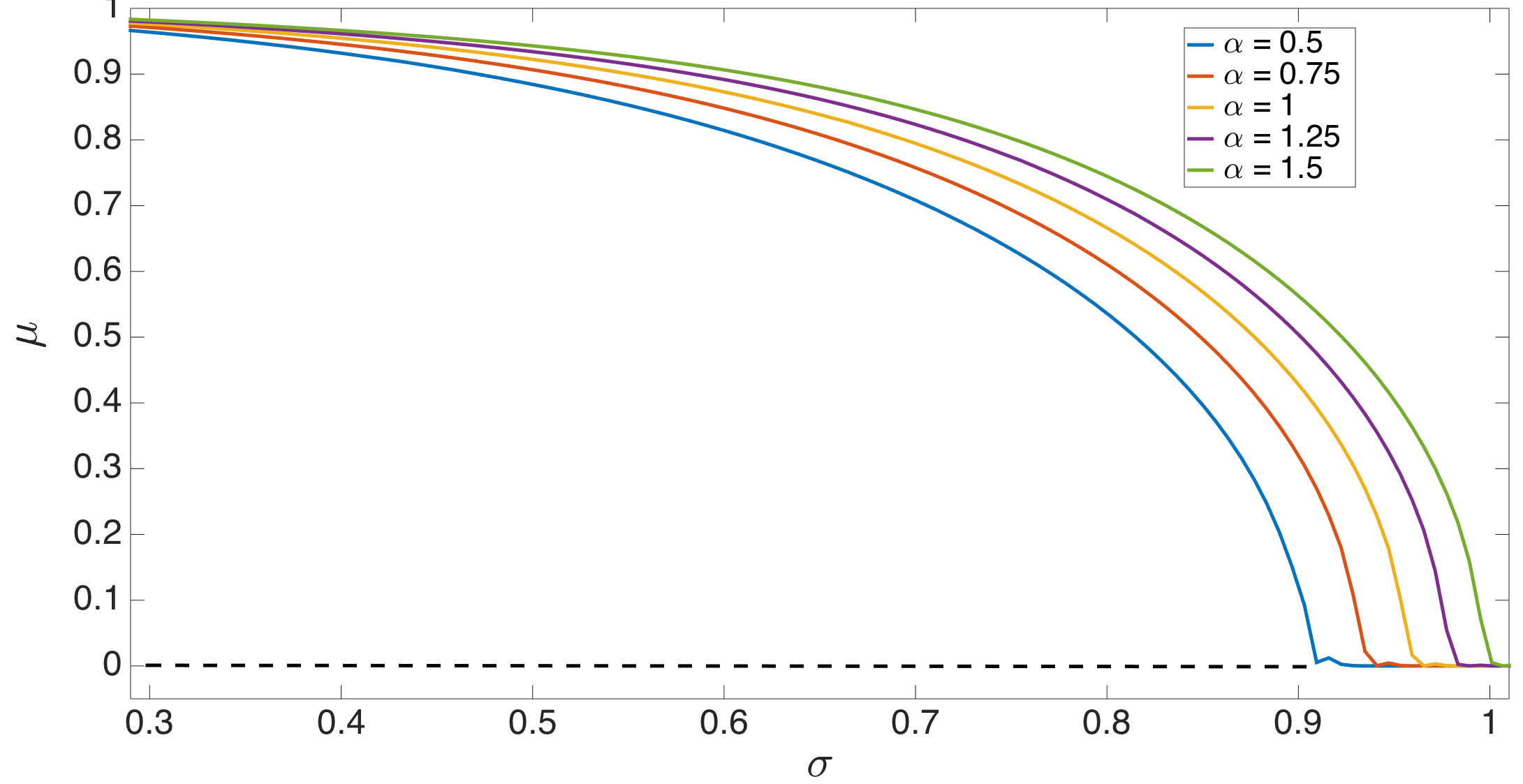}}
 \caption{\footnotesize Uncontrolled system. {a) Stable (solid) and unstable (dashed) zero mean steady state solutions for $\alpha=0.5$ (uni-modal, top) and $\alpha=1.5$ (bi-modal, bottom). b) Eigenvalues of $L_{loc}$ (o) and $L$ (${\bm{\times}}$)} for a typical zero mean case. The first eigenvalue $(= 0)$ is omitted. Notice that alternating eigenvalues are same for both operators. The arrows indicate the direction of motion of the other (`moving') eigenvalues of $L$ as $\sigma$ is reduced. The rightmost eigenvalue of $L$ reaches 0 at $\sigma=\sigma_c(\alpha)$ with non-zero speed. c) Non-zero mean solutions  d) The $\mu>0$ branch (solid) bifurcating from $\mu=0$ solution (dashed) via a supercritical pitchfork bifurcation as $\sigma$ occurs is reduced below $\sigma_c(\alpha)$. }
 \label{fig:un12}
 \end{figure*}

\paragraph*{\textbf{Numerical Results:}}
We use Chebfun \cite{driscoll2014chebfun} to perform all computations. The non-zero mean steady state solutions to Eq. (\ref{eq:kinetichom}) are computed using a simple fixed point iteration for $\mu$. The solutions are shown in Figure \ref{fig:un12}. The supercritical pitchfork bifurcation that occurs as $\sigma$ is reduced below critical value $\sigma_c(\alpha)$, is shown for a range of $\alpha$ values. To evaluate $h(\lambda)$ in Eq. (\ref{eq:eval3}), we compute the spectrum of $L_{loc}$ for $\mu=0$. The odd-numbered eigenfunctions are even functions of $x$, and hence $\langle \xi_{2k+1},g_1\rangle=0$. Therefore, eigenvalues $\lambda_{2k+1}$ of $L_{loc}$ are also eigenvalues of $L$, and the eigenvalues $\lambda_{2k}$ are moving eigenvalues. We find that at $\sigma=\sigma_c(\alpha)$, $h(0)=0$. Hence, as $\sigma$ is decreased below $\sigma_c(\alpha)$, the least stable eigenvalue $\lambda_2$ of $L_{loc}$ moves to the positive real axis due to the effect of the nonlocal term, resulting in instability of the zero mean solution.

\section{MFG Formulation}
In this section we describe a MFG formulation for homogeneous equation Eq. (\ref{eq:kinetichom}). The velocity of $i$th agent evolves via the following equation (compare with Eq (\ref{eq:unc_x}))
\begin{eqnarray}
dx_i(t)=a(x_i)dt+u_i(t)dt+\sigma d\omega_i(t),
\end{eqnarray}
where $u_i$ is the optimal control. Let $F[x_i,x_{-i}](t)\equiv (x_i-\dfrac{1}{N-1} \sum_{j\neq i}x_j)^2$, and $\beta=\dfrac{1}{r\sigma^2} $, where \kb{$r>0$} is the control \kb{cost or penalty}. Here $x_{-i}\equiv \{x_1,x_2,\dots,x_{i-1},x_{i+1},\dots\}$. Then the $i$th agent is minimizing the following long time average cost
\begin{eqnarray}
J=\limsup_{T\rightarrow\infty}\frac{1}{T}  \hspace{-0.1cm} \left[ \hspace{-0.05cm} \int_0^T\hspace{-0.25cm} \beta F[x_i,x_{-i}](t)+\frac{1}{2\sigma^2}u_i(t)^2 \right] dt\nonumber,
\end{eqnarray}
that depends on states of all other agents. 

To derive the MFG equations (recall Fig. \ref{fig:mfg}), we \kb{rewrite} the single-agent cost in terms of  $\hat F(x_i,t)$, the unknown coupling function with dependence on $x_i$ only
\begin{eqnarray}
J=\limsup_{T\rightarrow\infty}\frac{1}{T} \left[ \int_0^T \beta \hat F(x_i,t) +\frac{1}{2\sigma^2}u_i(t)^2 \right] dt.\nonumber
\end{eqnarray}

The resulting single agent HJB equation \cite{borkar2006ergodic} is 
\begin{flalign}
\partial_tv_i(x,t)=&c-\beta\hat F(x,t)-a(x)\partial_xv_i(x,t)\nonumber\\
+&\dfrac{\sigma^2}{2}(\partial_xv_i(x,t))^2 - \dfrac{\sigma^2}{2}\partial_{xx}v_i(x,t),
\end{flalign}
where $v_i(x,t)$ is the single-agent relative value function, $c$ is the minimum average cost, and $u_i(x,t)=-\sigma^2\partial_xv_i(x_i,t)$ given in feedback form. \PG{Note that the HJB equation is well-posed backward in time}. The self-consistency principle yields the expression for $\hat F$ in terms of agent density $f(x,t)$ (in the limit $N\rightarrow\infty$):
\begin{eqnarray}
\hat{F}[f](x,t)=\left[\int(x-y)f(y,t)dy\right]^2.
\end{eqnarray}

Hence, the following set of FP-HJB MFG equations govern the density and value function evolution:
\begin{align}
\partial_tf(x,t)+&\partial_x\left[\left(a(x)-\sigma^2\partial_xv(x,t)\right)f(x,t)\right]=\frac{\sigma^2}{2}\partial_{xx}f(x,t),\label{eq:fp1}\\
\partial_tv(x,t)&=c-\beta\hat{F}[f](x,t)-a(x)\partial_xv(x,t)\nonumber\\&+\frac{\sigma^2}{2}(\partial_xv(x,t))^2-\frac{\sigma^2}{2}\partial_{xx}v(x,t)\label{eq:hjb1}.
\end{align}

\kb{The unique invariant density satisfying Eq. (\ref{eq:fp1}) is}
\begin{eqnarray}
f_{\infty}(x)=\dfrac{1}{Z}\exp(-\frac{2}{\sigma^2} (U(x)+\sigma^2v_{\infty}(x))).\label{eq:finf}
\end{eqnarray}
Inserting this expression into Eq. (\ref{eq:hjb1}), and using the Cole-Hopf transformation \cite{todorov2009eigenfunction} $\phi(x)=\exp(-v_{\infty}(x))$, results in the following nonlinear nonlocal eigenvalue problem for $\phi(x)$
\begin{flalign}
c\phi(x)&=\beta\{x-\int y\exp(\frac{-2}{\sigma^2}U(y))\phi^2(y)dy\}^2\phi(x)\nonumber\\
&-a(x)\phi(x)-\frac{\sigma^2}{2}\partial_{xx}\phi(x),\label{eq:noneig}
\end{flalign}
with the constraint $\int \exp(\frac{-2}{\sigma^2}U(y))\phi^2(y) dy=1$ to ensure \kb{normalization of $f_{\infty}$}. The ground state of this problem yields the desired steady state solutions, with corresponding eigenvalue being the minimum cost $c$.
\vspace{-0.5cm}
\subsection{Stability Analysis}
In this section we extend the resolvent based analysis from section \ref{sec:stab_fp} to the MFG system, and find conditions for \emph{closed-loop} stability of an arbitrary steady state $(f_{\infty}(x),v_{\infty}(x))$ to an initial perturbation in density. We consider mass preserving perturbations in density of the form $f(x,t)=f_{\infty}(x)(1+\epsilon\tilde f(x,t))$, i.e., the initial conditions satisfy $\int f_{\infty}(x)\tilde f(x,0)dx=0$. The perturbed value function is taken to be of the form $v(x,t)=v_{\infty}(x)+\epsilon\tilde v(x,t)$. A given steady state is called linearly stable if any perturbation to the density decays to zero under the action of the control, where both the density and control evolution are computed using linearized MFG equations.

Linearization of MFG equations (\ref{eq:fp1},\ref{eq:hjb1}) yields the nonlocal system
\vspace{-0.25cm}
\begin{eqnarray}
\begin{bmatrix}
    \partial_t{\tilde{f}}(x,t)\\
    \partial_t{\tilde{v}}(x,t)
\end{bmatrix}
=
L^{FB}_{loc}
\begin{bmatrix}
    \tilde{f}(x,t)\\
    \tilde{v}(x,t)
\end{bmatrix},\label{eq:FB}
\end{eqnarray}
where
$L^{FB}=L^{FB}_{loc}+L^{FB}_{nonloc}$,
\begin{eqnarray}
L^{FB}_{loc}
=
\begin{bmatrix}
    L_{loc} &  2 L_{loc}  \\
   0  & -L_{loc}
     \end{bmatrix},\nonumber
L^{FB}_{nonloc}
=
\begin{bmatrix}
    0 & 0  \\
   2\beta s_1\langle g_1,.\rangle  & 0
     \end{bmatrix},\nonumber
\end{eqnarray}

$s_1(x)=x-\mu, g_1(x)=x$, $L_{loc}=-\partial_x(\hat U)\partial_x+\dfrac{\sigma^2}{2}\partial_{xx}$, with eigenvalue/eigenfunction pairs denoted by $\{\lambda_i,\xi_i\}$, and $\hat U(x)= U(x)+\sigma^2v^{\infty}(x)$ in analogy with the definition of $L_{loc}$ in Section \ref{sec:uncontrolled}. In addition to the Hilbert space $\mathbb{H}=L^2(\mathbb{R},f_{\infty}dx)$ and $R_{\lambda}$ as defined earlier, we also consider a subspace $\mathbb{\bar H}=\{f\in\mathbb{H}|\langle f,1\rangle=0\}$.
\subsubsection{Eigenspectrum of the linearized forward-backward operator}\label{subsec:spec}
We start off by noting that the characteristic equation of $L^{FB}_{loc}$ is $(L_{loc}-\lambda I)(L_{loc}+\lambda I)=0$. Hence, its eigenvalues are $\cup_{i\in\mathbb{N}}\{\pm\lambda_i\}$. Now consider the eigenvalue problem for $L^{FB}$ with eigenvalue $\lambda$ and eigenfunction $[w_f(x)\,\,w_v(x)]^T$:
\begin{eqnarray}
\lambda\begin{bmatrix}
    w_f\\
    w_v
\end{bmatrix}
=
\begin{bmatrix}
    L_{loc} w_f+2 L_{loc}w_v\label{eq:mfgeval1}\\
    2\beta s_1\langle g_1,w_f\rangle -L_{loc}w_v
\end{bmatrix}.
\end{eqnarray}

Assuming  $\lambda\not\in\cup_{i\in\mathbb{N}}\{\pm\lambda_i\}$, $R_{\lambda}$ and $R_{-\lambda}$ are well defined. The second equation of Eq. (\ref{eq:mfgeval1}) gives
\begin{eqnarray}
w_v=2\beta R_{-\lambda} s_1\langle g_1,w_f\rangle.\nonumber
\end{eqnarray}
Substituting this expression in the first equation of Eq. (\ref{eq:mfgeval1}), and re-arranging, 
\begin{eqnarray}
w_f=-4\beta R_{\lambda}L_{loc}R_{-\lambda} s_1\langle g_1,w_f\rangle.\label{eq:mfgeval1a}
\end{eqnarray}
Taking the inner product of the above equation with $g_1$,
\begin{eqnarray}
\langle g_1,w_f\rangle(1+4\beta\langle g_1, R_{\lambda}L_{loc}R_{-\lambda}s_1\rangle)=0.\nonumber
\end{eqnarray}

The eigenvalue equation for the $\langle g_1,w_f\rangle\neq0$ case for moving eigenvalues (as in Section \ref{sec:stab_fp}) is
\begin{eqnarray}
h({\lambda})\equiv 1+4\beta\langle g_1, R_{\lambda}L_{loc}R_{-\lambda}s_1\rangle=0.\label{eq:mfgeval2}
\end{eqnarray}
Using the definition of resolvent in Eq. (\ref{eq:mfgeval2}),
\begin{eqnarray}
h({\lambda})=1+4\beta\sum_{i=2}^{\infty}\lambda_i\dfrac{\langle g_1,\xi_i\rangle\langle s_1,\xi_i\rangle}{\lambda_i^2-\lambda^2},\label{eq:mfgeval3}
\end{eqnarray}
and hence,
\begin{eqnarray}
h({\lambda})=1+4\beta\sum_{i=2}^{\infty}\lambda_i\dfrac{\langle x,\xi_i\rangle^2}{\lambda_i^2-\lambda^2}.\label{eq:mfgeval4}
\end{eqnarray}
Since Eq. (\ref{eq:mfgeval4}) is Herglotz in $\lambda^2$, this implies that the eigenvalues come in pairs, either real or purely imaginary. \PG{Let $\omega\equiv h(0)=1+4\beta\sum_{i=2}^{\infty}\dfrac{\langle x,\xi_i\rangle^2}{\lambda_i}$.
\begin{lemma}\label{lem:mono} Consider the eigenvalue equation $h(\lambda)=0$ for moving eigenvalues.
\begin{itemize}
\item[(i)]  If $\langle x,\xi_i\rangle \neq 0$ for all $i\geq 2$, then there exists a pair of real roots $\pm\delta_i$ for each $i\geq 2$, such that $\lambda_{i+1}<\delta_i<\lambda_i$.
\item[(ii)] Recall that $\lambda_1=0$. If $\langle x,\xi_2\rangle \neq 0$ and $\omega>0$, there exists a pair of real roots $\pm\delta_1$, such that $\lambda_2<\delta_1<0$.
\item[(iii)]If $\langle x,\xi_2\rangle \neq 0$ and $\omega<0$, there exists a pair of purely imaginary roots $\pm i\gamma$.
\end{itemize}
\end{lemma}
\begin{proof}
\begin{itemize}
\item[(i)]  Consider the interval $I_i=(\lambda_{i+1},\lambda_i)$. As $\lambda\rightarrow\lambda_i^-$, $h(\lambda)\rightarrow \infty$, and as $\lambda\rightarrow\lambda_{i+1}^+$, $h(\lambda)\rightarrow -\infty$. It is easy to check that $h(\lambda)$ is monotonic in $I_i$. By intermediate value theorem, a root $\delta_i$ exists in $I_i$, and by the monotonicity property, it is unique. The result for $-\delta_i$ follows by symmetry.
\item[(ii)] Consider the interval $I_1=(\lambda_2,0)$. Note that as $\lambda\rightarrow\lambda_{2}^+$, $h(\lambda)\rightarrow -\infty$, and as $\lambda\rightarrow 0^-$, $h(\lambda)\rightarrow \omega$. Hence, if $w>0$, arguments similar to those in part (i) yield the existence of a real root $\delta_1$ between $\lambda_2$ and $0$.
\item[(iii)] Consider the function $h(i\gamma)$ for real $\gamma>0$. Clearly, $h$ is monotonic in this interval. Furthermore, as $\gamma\rightarrow \infty$, $h(i\gamma)\rightarrow 0$, and as $\gamma\rightarrow 0^+$, $h(i\gamma)\rightarrow \omega$. By arguments similar to those in part (i), $\omega<0$ implies that there is a unique root $i\gamma$ of $h$.
\end{itemize}
\end{proof}}
\vspace{-0.3cm}
\subsubsection{Contraction analysis of the linearized forward-backward operator}
Since the MFG system has a forward-backward nature, spectral information alone is insufficient to derive conclusions about the stability of steady state solutions. \PG{A contraction analysis is therefore adopted following Refs. \onlinecite{yin2012synchronization,huang2007large}. Consider the linear dynamical system given by Eq. (\ref{eq:FB}), with initial perturbation in density $f(x,0)=f_{\infty}(1+\epsilon\tilde f(x,0))$. Assuming that $\tilde v(x,T)\rightarrow 0 $ as $T\rightarrow \infty$, the conditions for existence of a unique solution satisfying this assumption are derived. These conditions also provide a stability criterion.}
Integrating the $\tilde{v}$ equation in  Eq. (\ref{eq:FB}) from $t$ to $T$, 
\vspace{-0.375cm}
\begin{flalign*}
\tilde{v}(x,T)&=e^{-L_{loc}(T-t)}\tilde{v}(x,t)\\&+2\beta\int_t^Te^{-L_{loc}(T-s)}s_1(x)\langle g_1,\tilde{f}(.,s)\rangle ds.
\end{flalign*}
Taking the limit $T\rightarrow\infty$,

\begin{flalign}
\tilde{v}(x,t)=-2\beta e^{-L_{loc}t} \hspace{-0.15cm} \int_t^{\infty} \hspace{-0.25cm}  e^{L_{loc}s}s_1(x)\langle g_1(.),\tilde{f}(.,s)\rangle ds.\label{eq:vper}
\end{flalign}

\noindent \kb{Substituting above equation in the $\tilde{f}$ equation},
\begin{flalign}
\partial_t\tilde{f}(x,t)&=L_{loc}\tilde{f}(x,t)\nonumber\\&-4\beta L_{loc}e^{-L_{loc}t}\int_t^{\infty}e^{L_{loc}s}s_1(x)\langle g_1(.),\tilde{f}(.,s)\rangle ds.
\end{flalign}
\kb{Integrating from $0$ to $t$ yields the fixed point equation},
\begin{flalign}
\tilde{f}(x,t)=e^{L_{loc}t}\tilde{f}(x,0)+M\tilde{f}(x,t),\label{eq:fp_M}
\end{flalign}
where the operator $M$ acting on $\tilde{f}(x,t)$ is defined as
\begin{widetext}
\begin{flalign}
M\tilde{f}(x,t)=-4\beta e^{L_{loc}t}\int_{r=0}^{r=t} e^{-L_{loc}r}L_{loc}e^{-L_{loc}r}\int_{s=r}^{s=\infty}e^{L_{loc}s}s_1(x)\langle g_1(.),\tilde{f}(.,s)\rangle dsdr. \label{eq:Mtime}
\end{flalign}
\end{widetext}
Applying the Laplace transform in time to Eq. (\ref{eq:Mtime}),
\begin{equation}
\hat{M}(\lambda)=-4\beta R_{\lambda}L_{loc}R_{-\lambda}s_1\langle g_1,.\rangle. \label{eq:Mfreq}
\end{equation}

\begin{figure*}
 	\subfloat[]{\includegraphics[width=2.35in,height=2in]{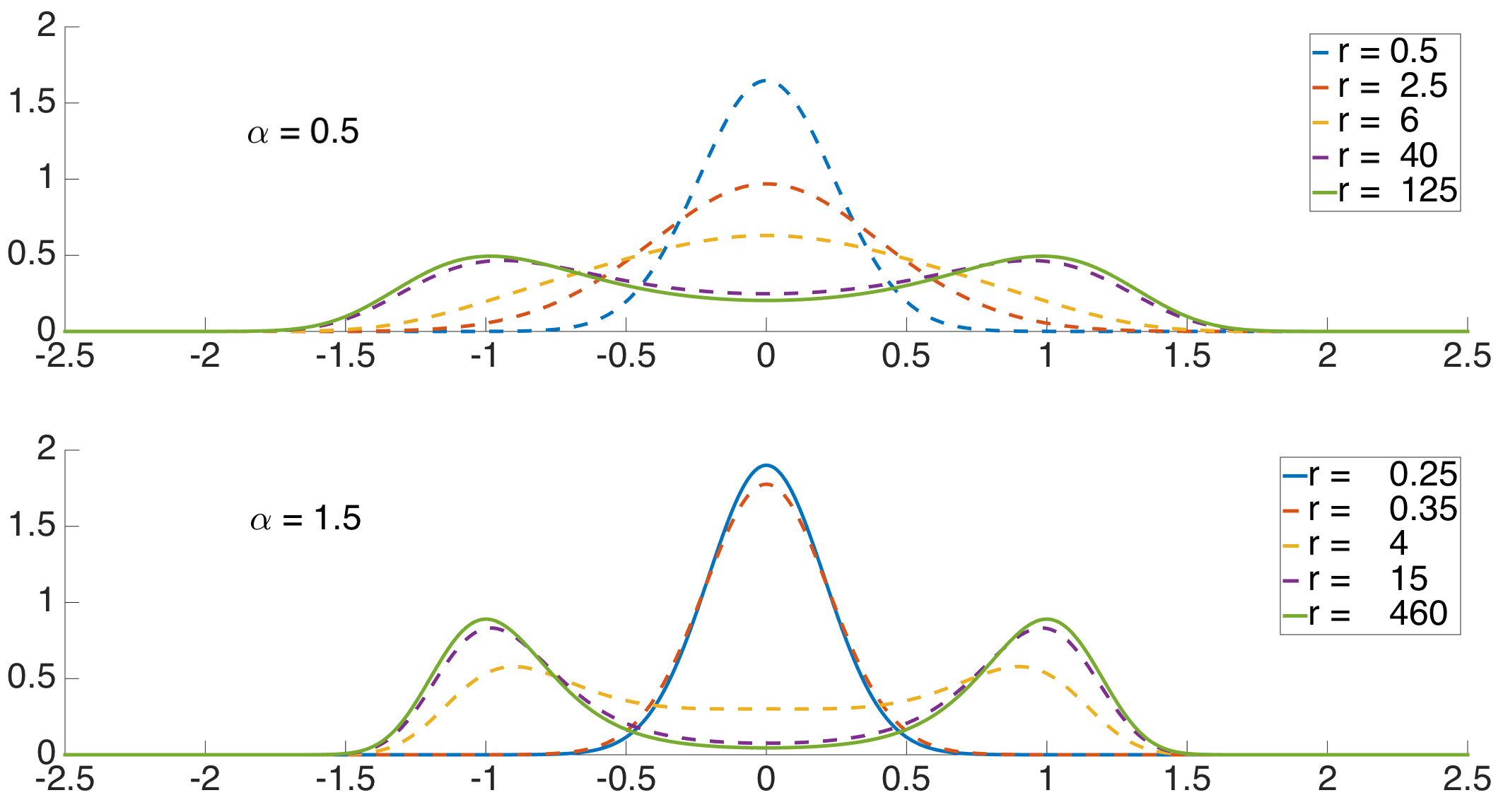}	}
  \subfloat[]{\includegraphics[width=2.35in]{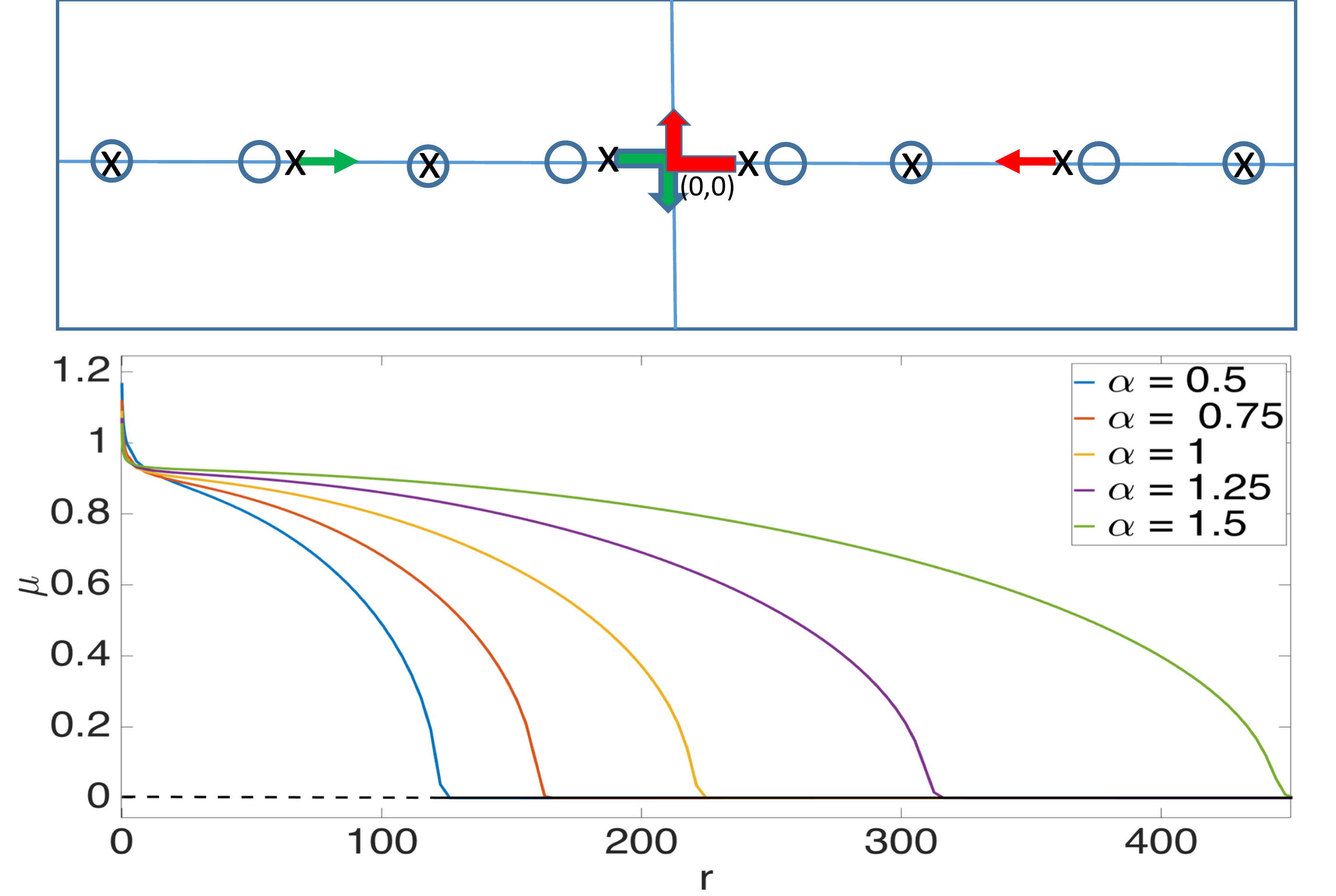}	}
 	\subfloat[]{\includegraphics[width=2.35in,height=2in]{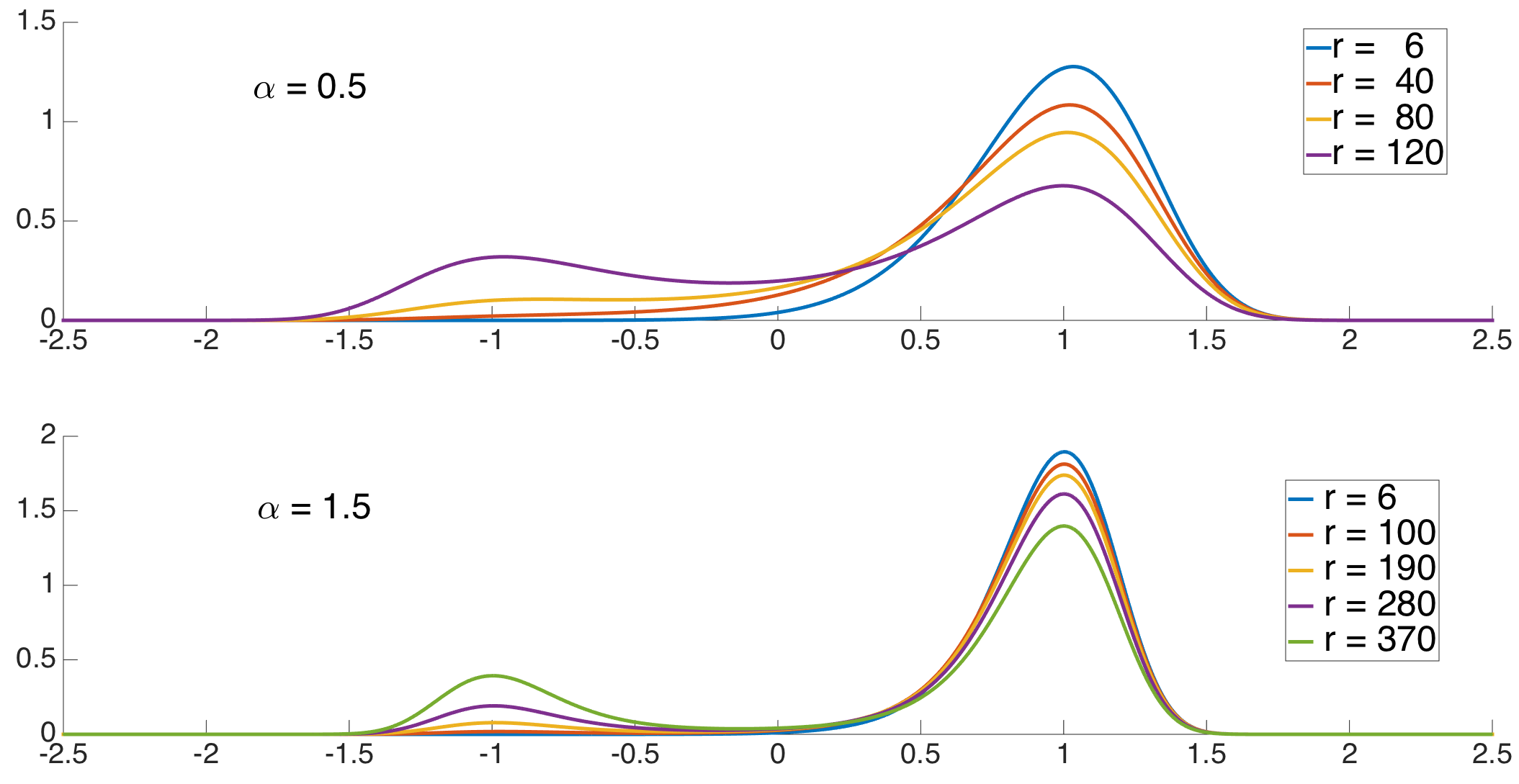}	}
\caption{\footnotesize {The MFG system with $\sigma=0.5$. a) Zero mean MFG steady state densities for $\alpha=0.5$ (top) and $\alpha=1.5$ (bottom) for various control penalty values. b) (Top) Eigenvalues of $L^{FB}_{loc}$ (o) and $L^{FB}$ (${\bm{\times}}$)} for a typical zero mean case. The twin zero eigenvalues of $L^{FB}$ are omitted. The arrows indicate the direction of motion of the `moving' eigenvalues of $L^{FB}$ as $r$ is reduced starting from $r>r_{sup}(\alpha,\sigma)$. \PG{Note that the assumption in Lemma \ref{lem:mono}(i) is violated in this particular case due to the symmetric nature of the self-propulsion term, and hence, only alternating eigenvalues are actually `moving'.} The pair of eigenvalues of $L^{FB}$ closest to imaginary axis, $\pm\delta_1$, reaches $0$ at $r=r_{sup}$, and moves up/down the imaginary axis for $r<r_{sup}$. (Bottom) The $\mu>0$ branch (solid) bifurcating from $\mu=0$ solution (dashed) via a supercritical pitchfork bifurcation as $r$ is reduced below $r_{sup}$. c) Non-zero mean MFG steady state densities on the supercritical branch. }
 \label{fig:mfgfp1}
 \end{figure*}

The operator norm $\|M\|$ is given by 
\begin{flalign}
\|M\|=&\sup_{\lambda\in \mathbb{I}}\sup_{\|\tilde {f}\|=1}\|\hat{M}(\lambda)\tilde{f}\|\label{eq:Mnorm},\\
=&4\beta \sup_{\lambda\in \mathbb{I}}\sup_{\|\tilde f\|=1} \| \sum_{i=2}^{\infty}\dfrac{\lambda_i\langle s_1,\xi_i\rangle\langle g_1,\tilde f\rangle}{\lambda_i^2-\lambda^2}\xi_i\|,\nonumber
\end{flalign}
\begin{flalign}
=&4\beta \sup_{\lambda\in \mathbb{I}}\sup_{\|\tilde f\|=1} \sqrt{\sum_{i=2}^{\infty}\left[\dfrac{\lambda_i\langle s_1,\xi_i\rangle\langle g_1,\tilde f\rangle}{\lambda_i^2-\lambda^2}\right]^2},\nonumber\\
=&4\beta \|g_1\| \sqrt{\sum_{i=2}^{\infty}\dfrac{\langle s_1,\xi_i\rangle^2}{\lambda_i^2}}=4\beta \|x\| \sqrt{\sum_{i=2}^{\infty}\dfrac{\langle x,\xi_i\rangle^2}{\lambda_i^2}}.
\end{flalign}
Lemma \ref{lem:contract} proved next implies that $\|M\|<1$ is a sufficient condition for a steady state $(f_{\infty}(x),v_{\infty}(x))$ of the nonlinear MFG system Eqs. (\ref{eq:fp1},\ref{eq:hjb1}) to be linearly stable to density perturbations.

\begin{lemma}\label{lem:contract} Consider the initial value problem for the linearized system in Eqs. \ref{eq:FB}, with mass-preserving initial condition $\tilde f(x,0)$ i.e., $\int f_{\infty}(x)\tilde f(x,0)dx=0$. If the operator $M$ is a contraction (i.e., $\|M\|<1$), then the perturbation in density, $\tilde f(.,t)$, decays to $0$ as $t\rightarrow \infty$. Moreover, $\tilde v(.,t)$ also decays to $0$ as $t\rightarrow\infty$.
\end{lemma}
\begin{proof}
If $M$ is a contraction, then we can (formally) invert the Eq. (\ref{eq:fp_M}), and write the unique solution
\begin{eqnarray}
\tilde{f}(x,t)=(\mathbf{I}-M)^{-1}e^{L_{loc}t}\tilde{f}(x,0)\nonumber\\=(\mathbf{I}+M+M^2+\dots)e^{L_{loc}t}\tilde{f}(x,0).
\end{eqnarray}
We note that mass conservation property is equivalent to $\langle \tilde f(x,0),1\rangle=0$, i.e. $\tilde f(x,0) \in \mathbb{\bar{H}}$.  Recall that $L_{loc}$ restricted to $\mathbb{\bar{H}}$ is a self-adjoint operator with negative eigenvalues $\lambda_i, i=2,3,\dots$. Then, $\lim_{t\rightarrow\infty}\|e^{L_{loc}t}\|_{\mathbb{\bar{H}}}=\lim_{t\rightarrow\infty}e^{\lambda_{2}t}=0$. This proves the decay of $\tilde f(.,t)$. \PG{The corresponding result for $\tilde v(.,t)$ is obtained by inserting the expression for $\tilde f(.,t)$ into Eq. (\ref{eq:vper}).}
\end{proof}

Now consider a case where eigenvalue equation in Eq. (\ref{eq:mfgeval1a}) has a pair of  purely imaginary roots $\pm i\gamma (\neq 0)$. Then there is a eigenfunction $z_f$ s.t.
\begin{eqnarray*}
z_f=-4\beta R_{i\gamma}L_{loc}R_{-i\gamma}s_1\langle g_1,z_f\rangle\\
=\hat{M}(i\gamma)z_f,
\end{eqnarray*}
by noting Eq. (\ref{eq:Mfreq}). But this implies that norm of $\hat{M}$ is at least $1$, hence it is not a contraction. This implies that a necessary condition for $M$ to be a contraction is the absence of non-zero spectra of $L^{FB}$ on the imaginary axis.

 \subsection{Numerical Results}
 \begin{figure*}[t!]
 \centering
	 	\subfloat[]{\includegraphics[width=2.35in,height=1.25in]{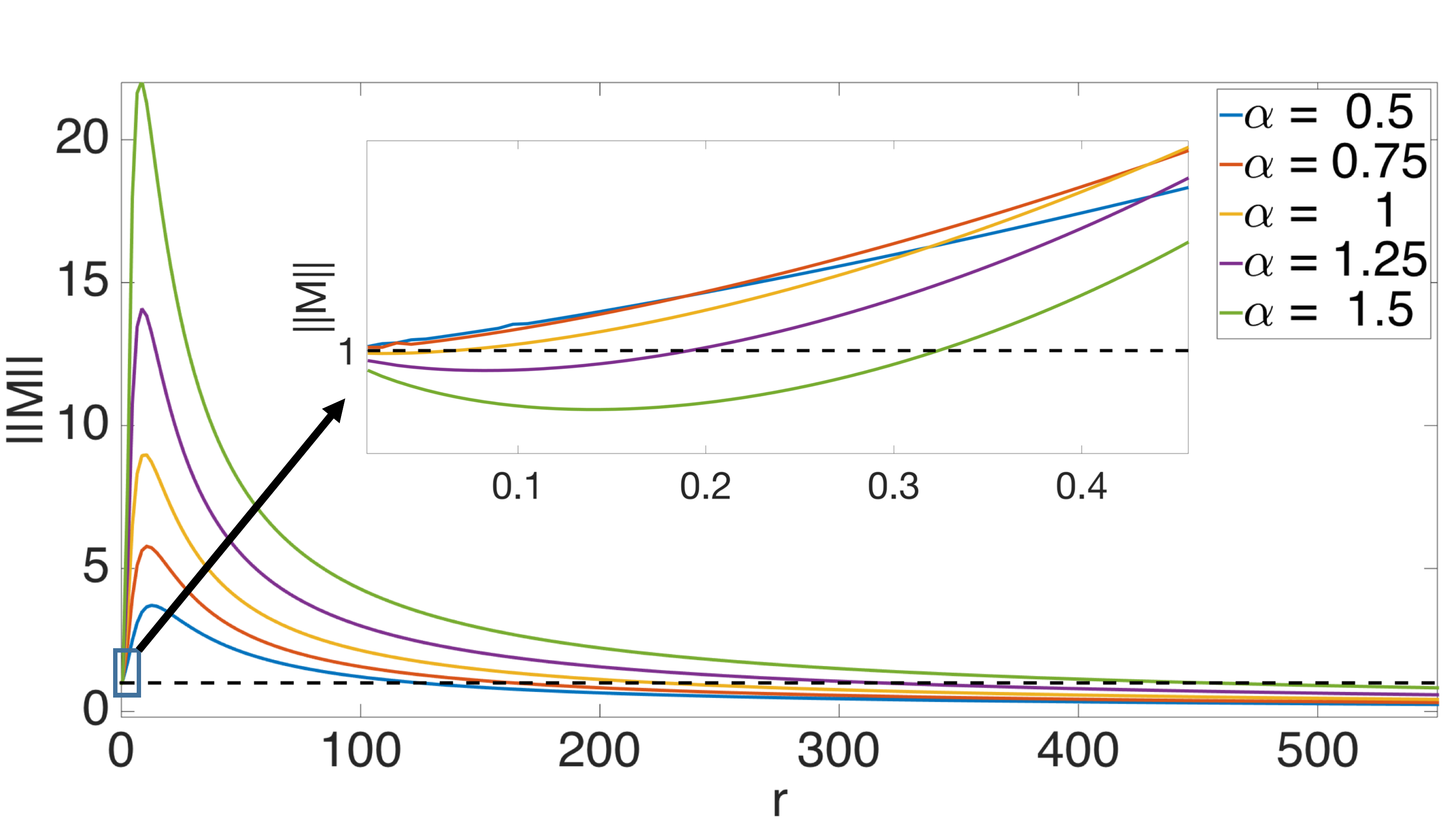}}
	 	\subfloat[]{\includegraphics[width=2.35in,height=1.25in]{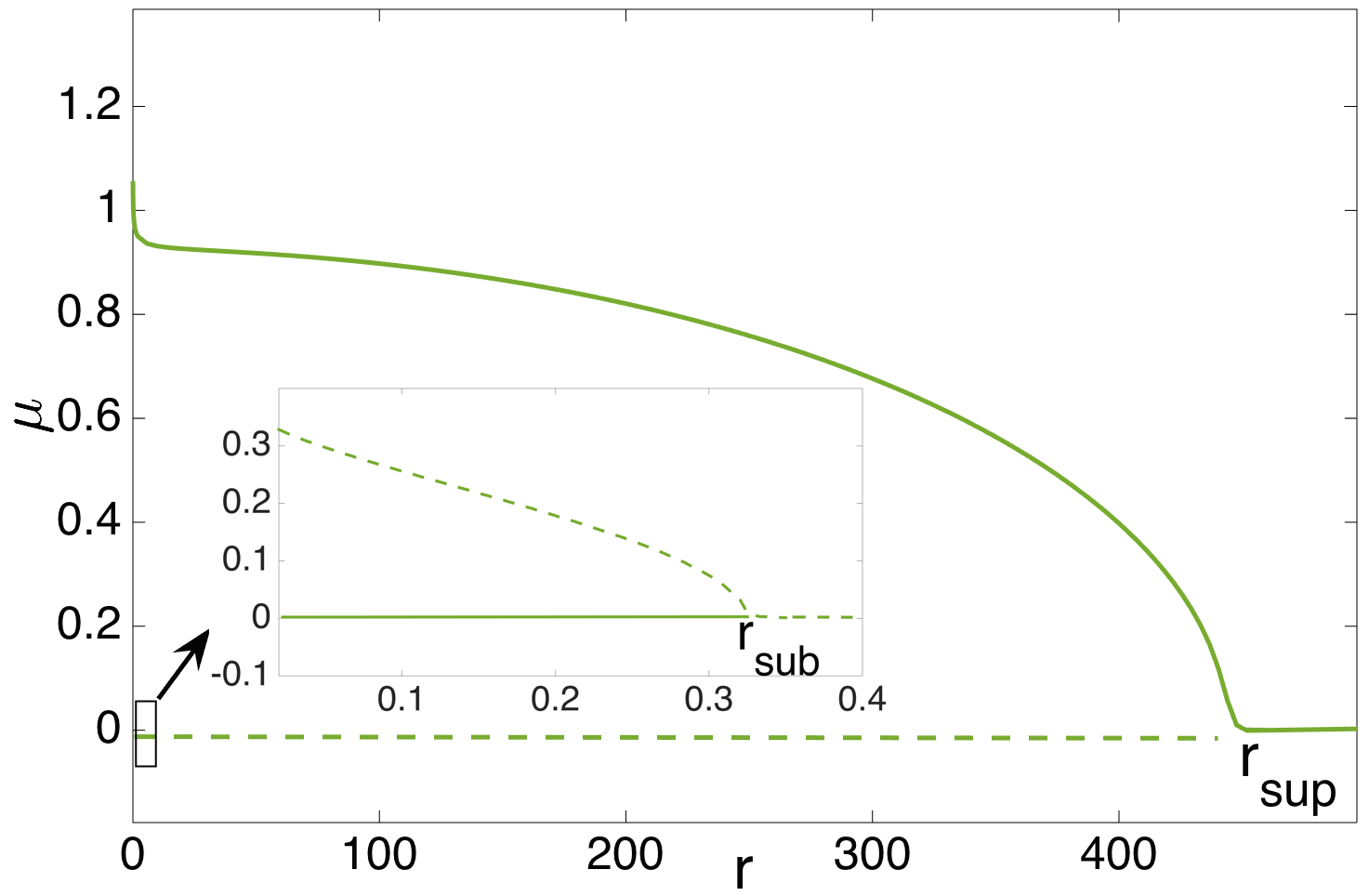}}
	 	\subfloat[]{\includegraphics[width=2.35in, height=1in]{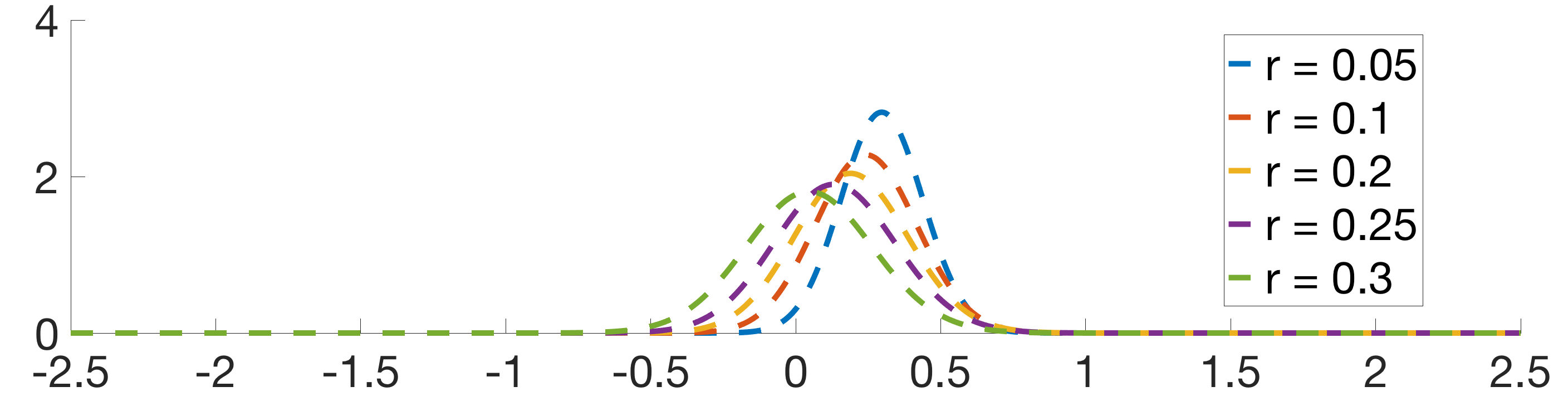}}
		\caption{\footnotesize{The MFG system with $\sigma=0.5$. a) The norm of operator $M$ for zero mean steady state as control penalty is varied, for various $\alpha$. b). The bifurcation diagram for $\alpha=1.5$, showing supercritical and subcritical (inset) bifurcations . Only the $\mu>0$ non-zero mean branches are shown. c) Non-zero mean MFG steady state densities on the subcritical branch for $\alpha=1.5$. }}
 \label{fig:mfgfp2}
 \end{figure*}
 \PG{Recall that in the MFG problem described by Eqs. (\ref{eq:fp1},\ref{eq:hjb1}), the representative agent is minimizing a weighted sum of two costs: one penalizes deviation of its velocity from the mean velocity of the agent population, and the other penalizes the control action. In this section, we compute fixed points, and identify phase transitions of this system of equations as the problem parameters are varied.} Rather than solving the resulting constrained nonlinear eigenvalue problem \ref{eq:noneig} directly, we use an iterative algorithm to compute steady state solutions of the MFG system.

We note that the coupling term $\hat{F}[f](x,t)$ evaluated at any steady state density $f_{\infty}$ is  $\hat{F}[f_{\infty}](x)=(x-\mu)^2$, where $\mu=\int yf_{\infty}(y)dy$. Again using Cole-Hopf transformation $\phi(x)=\exp(-v(x))$ on the HJB equation leads to a \textbf{linear} eigenvalue problem in $\phi$:
\begin{eqnarray}
c\phi(x)=\mathcal{L}[\phi](x),\label{eq:lineig}
\end{eqnarray}
where $\mathcal{L}[\phi]=\beta(x-\mu)^2\phi(x)-a(x)\partial_x\phi(x)-\frac{\sigma^2}{2}\partial_{xx}\phi(x)$. 
We solve Eq. (\ref{eq:lineig})  iteratively along with Eq. (\ref{eq:finf}) to find zero mean MFG steady states ($f_{\infty},v_{\infty})$ for a range of $r$, keeping $\sigma$ and $\alpha$ fixed (See Fig. \ref{fig:mfgfp1}).These solutions are stable (i.e., $\|M\|<1$) for large $r$, \PG{implying that when control is expensive, the agents use minimal control action. The resulting steady state distribution is bi-modal due to dominance of the self-propulsion force, and dispersion via noise}.  

These zero mean solutions lose stability (i.e., $\|M\|>1$) via a supercritical bifurcation as $r$ is reduced below a critical value $r_{sup}$. The Eq. (\ref{eq:mfgeval4}) for moving eigenvalues of $L^{FB}$ has a double zero root at $r=r_{sup}$, and a pair of purely imaginary roots emerges as $r$ is reduced below $r_{sup}$. This implies that the pair of symmetric eigenvalues of $L^{FB}$ closest to the imaginary axis reaches $0$ at the critical parameter, and then moves up/down the imaginary axis. The stable non-zero mean MFG steady state solutions on the supercritical branch are computed by combining fixed point iteration in $\mu$ with a continuation step.  \PG{This bifurcation provides a MFG interpretation to the pitchfork bifurcation observed in the uncontrolled system, i.e., cheaper control makes it economical to compensate for noise. Hence, the agents apply larger control action to flock together (and reduce the cost of deviation from the population mean), resulting in symmetry breaking non-zero mean solutions.}

When noise strength $\sigma$ is fixed below a critical value, the zero mean solution branch undergoes a \emph{subcritical} bifurcation as control penalty $r$ is further reduced, i.e, at $r=r_{sub}< r_{sup}$ (See Fig. \ref{fig:mfgfp2}). The corresponding non-zero mean solutions were computed using bisection method. This bifurcation is not seen in the uncontrolled system. For instance, when ($\sigma=0.5,  \alpha=1.5$), it results in creation of uni-modal stable zero mean solutions in the case of cheap control, $r<r_{sub}$, as compared to the bi-modal stable zero mean solution that exist for expensive control, $r>r_{sup}$. Hence, we conclude that for $r<r_{sub}$, the control is cheap enough to counteract the intrinsic dynamics, and make zero mean uni-modal solution stable.
\vspace{-0.3cm}
\section{Conclusions}
We have presented a MFG formulation for homogeneous flocking of agents with gradient nonlinearity in their intrinsic dynamics. We have employed tools from theory of reaction-diffusion equations, and exploited the low rank nature of the nonlocal coupling term to study the linear stability of the MFG equations. The explicit formulae for verifying the stability of steady state solutions of the nonlocal forward-backward MFG system require relatively simple numerical computation of spectra of the local self-adjoint Fokker-Planck operators. The MFG system shows rich nonlinear behavior, such as supercritical and subcritical pitchfork bifurcations that result in wide range of collective behaviors, some of which are not present in the uncontrolled model. 

Much of the analysis in the current work can be generalized to higher dimensional state space for homogeneous flocking with self-propulsion, similar in spirit to the generalization\cite{barbaro2016phase} of one-dimensional \emph{uncontrolled} flocking model. Furthermore, the abstract results presented in this work apply to models other than homogeneous flocking, e.g. nonlocally coupled agents with arbitrary first order gradient dynamics. Extension to non-homogeneous flocking would be a natural next step; the resulting second-order dynamics could require more sophisticated tools \cite{alexander1990topological} for stability analysis. \PG{Implementation of the MFG control laws in an engineered large population system requires the control to be provided in a causal form. Algorithms that can learn the MFG laws can be used to convert the control laws obtained by solving the FP-HJB equations into an implementable form \cite{cardaliaguet2017learning}.}

The use of bifurcation and singularity theory to develop bio-inspired control and decision making algorithms for multi-agent systems has been explored recently\cite{leonard2014multi, srivastava2017bio,gray2015honeybee}. Our work adds to the toolbox for systematic analysis of collective behavior of non-cooperative dynamic agents via an inverse modeling approach. \PG{The qualitative and quantitative insight provided by the stability analysis can be exploited in \emph{mechanism design}, i.e., design of penalties or incentives to drive the population to asymptotic states with desirable characteristics. We believe that a systematic study of bifurcations in MFG models can lead to progress in tackling the grand challenge of designing or manipulating collective behavior of a large population of non-cooperative dynamic agents.}
\begin{acknowledgments}
We wish to thank the anonymous reviewers for their constructive comments, especially regarding Lemma \ref{lem:mono} and its proof.
\end{acknowledgments}
\end{document}